\documentclass[11pt,a4paper]{article}

\title{On the asymptotic distribution of the mean absolute deviation about the mean}
\author{%
  Johan Segers\\
  \small Universit\'e catholique de Louvain\\
  \small Institut de Statistique, Biostatistique et Sciences Actuarielles\\
  \small Voie du Roman Pays 20, B-1348 Louvain-la-Neuve, Belgium\\
  \small \texttt{johan.segers@uclouvain.be}%
}
\date{June 16, 2014}

\usepackage{amsmath}
\usepackage{amsfonts}
\usepackage{amssymb}
\usepackage{amsthm}
\usepackage{dsfont}
\usepackage{paralist}
\usepackage{natbib}

\newtheorem{proposition}{Proposition}[section]
\newtheorem{lemma}[proposition]{Lemma}

\newtheorem{corollary}[proposition]{Corollary}

\theoremstyle{remark}
\newtheorem{remark}[proposition]{Remark}

\numberwithin{equation}{section}

\newcommand{\diff}{\mathrm{d}}

\renewcommand{\Pr}{\mathbb{P}}
\newcommand{\E}{\mathbb{E}}
\newcommand{\expec}{\E}
\newcommand{\cov}{\operatorname{cov}}

\newcommand{\var}{\operatorname{var}}
\newcommand{\sign}{\operatorname{sign}}
\newcommand{\reals}{\mathds{R}}
\newcommand{\RR}{\reals}
\newcommand{\ZZ}{\mathds{Z}}
\newcommand{\ind}{\mathds{1}}
\newcommand{\1}{\ind}

\newcommand{\dto}{\rightsquigarrow}

\renewcommand{\le}{\leqslant}
\renewcommand{\ge}{\geqslant}

\newcommand{\floor}[1]{\left\lfloor{#1}\right\rfloor}

\newcommand{\MAD}{\theta}
\newcommand{\MADn}{\hat{\MAD}_n}
\newcommand{\MADm}{\tilde{\MAD}_n}

\newcommand{\abs}[1]{\left\lvert{#1}\right\rvert}

\begin{document}

\maketitle

\begin{abstract}
The mean absolute deviation about the mean is an alternative to the standard deviation for measuring dispersion in a sample or in a population.
For stationary, ergodic time series with a finite first moment, an asymptotic expansion for the sample mean absolute deviation is proposed. The expansion yields the asymptotic distribution of the sample mean absolute deviation under a wide range of settings, allowing for serial dependence or an infinite second moment.
\par\medskip
\noindent\emph{Key words:} central limit theorem; dispersion; ergodicity; regular variation; stable distribution; strong mixing.
\end{abstract}

\section{Introduction}
\label{sec:intro}

The mean absolute deviation of a sample $X_1, \ldots, X_n$ about the sample mean $\bar{X}_n = n^{-1} \sum_{i=1}^n X_i$ is given by the statistic
\begin{equation}
\label{eq:MADn}
  \MADn = \frac{1}{n} \sum_{i=1}^n \abs{X_i - \bar{X}_n}.
\end{equation}
If the random variables $X_i$ have common distribution function $F$ with finite mean $\mu = \int_\reals x \, \diff F(x)$, then $\MADn$ is an estimator of the mean absolute deviation
\begin{equation}
\label{eq:MAD}
  \MAD = \expec[ \abs{X_1 - \mu} ] = \int_{\RR} \abs{x - \mu} \, \diff F(x).
\end{equation}


The (sample) mean absolute deviation is an alternative to the standard deviation for measuring dispersion. Its advantages and drawbacks have been widely discussed in the literature. The standard deviation is motivated mainly from optimality results in the context of independent random sampling from the normal distribution, an analysis dating back to Fisher and even Laplace \citep{stigler:1973}. However, the mean absolute deviation may be more appropriate in case of departures from normality or in the presence of outliers \citep{tukey:1960, huber:1996}. It may also offer certain pedagogical advantages. For extensive discussions and comparisons, see \cite{pham-gia:hung:2001} and \cite{gorard:2005} and the references therein.

Because of the presence of the absolute value function, finding the asymptotic distribution of the sample mean absolute deviation is surprisingly challenging. In \citet[Section~2]{pollard:1989} and \citet[Example~19.25]{vdvaart:1998}, the exercise is put forward as a showcase for the power of empirical process theory. In a nutshell, their analysis is as follows. Let
\[
  \MADm = \frac{1}{n} \sum_{i=1}^n \abs{X_i - \mu}
\]
be the version of the sample mean absolute deviation that could be computed if the true mean, $\mu$, were known. Consider the dispersion function
\[
  D_F(u) = \int_{\RR} \abs{x - u} \, \diff F(x), \qquad u \in \reals.
\]
Clearly, $\MAD = D_F(\mu)$. Under independent random sampling and under the presence of finite second moments, asymptotic uniform equicontinuity of the empirical process $u \mapsto n^{-1/2} \sum_{i=1}^n \bigl( \abs{X_i - u} - D_F(u) \bigr)$ implies that
\begin{align}
\nonumber
  \sqrt{n} \bigl( \MADn - \MAD \bigr)
  &= \sqrt{n} \bigl( \MADn - D_F(\bar{X}_n) \bigr) + \sqrt{n} \bigl( D_F( \bar{X}_n ) - \MAD \bigr) \\
\label{eq:MADn:vdV}
  &= \sqrt{n} \bigl( \MADm - \MAD \bigr) + \sqrt{n} \bigl( D_F( \bar{X}_n ) - \MAD \bigr)
  + o_p(1), \qquad n \to \infty.
\end{align}
If $F$ is continuous at $u$, then $D_F$ is differentiable at $u$ with derivative $D_F'(u) = 2 F(u) - 1$ \citep[Theorem~1]{munoz:sanchez:1990}. By the delta method, it follows from \eqref{eq:MADn:vdV} that, if $F$ is continuous at $\mu$,
\begin{align}
\label{eq:MAD:empproc}
  \sqrt{n} \bigl( \MADn - \MAD \bigr)
  &= \sqrt{n} \bigl( \MADm - \MAD \bigr) 
  + \bigl(2 F(\mu) - 1 \bigr) \, \sqrt{n} ( \bar{X}_n - \mu ) + o_p(1) \\
\label{eq:MAD:asnorm}
  &\dto N(0, \sigma_{\MAD}^2), \qquad n \to \infty,
\end{align}
the arrow `$\dto$' denoting weak convergence. The asymptotic variance equals
\[
  \sigma_{\MAD}^2 = \var \bigl( \abs{X_1 - \mu} + (2 F(\mu) - 1) \, X_1 \bigr).
\]

The above proof is elegant and short. However, it rests on a body of advanced empirical process theory. Using direct arguments, \cite{babu:rao:1992} establish higher-order expansions for $\sqrt{n} \bigl( \MADn - \MAD \bigr)$ under the additional assumption that $F$ is H\"older continuous or even differentiable at $\mu$.

The results described so far are limited to independent random sampling from a distribution $F$ with a finite second moment, and, except for \eqref{eq:MADn:vdV}, without an atom at its mean $\mu$. In the literature, no results seem to be available on the asymptotic distribution of the sample mean absolute deviation in the case of serial dependence or when the second moment does not exist. Even in the case of independent random sampling from a distribution with finite second moment, the asymptotic distribution of \eqref{eq:MADn:vdV} in case $F$ has an atom at $\mu$ seems not to have been described yet.

The aim of this paper is to derive an asymptotic expansion of $\MADn - \MAD$ from first principles and under minimal assumptions (Section~\ref{sec:expansion}). The expansion yields the asymptotic distribution of the sample mean absolute deviation under a wide range of settings, including serial dependence, in the infinite-variance case, and without smoothness assumptions (Section~\ref{sec:dto}). Even in the case of independent random sampling from a distribution with finite second moment, the asymptotic distribution of $\sqrt{n} (\MADn - \MAD)$ is found to be non-Gaussian in case $F$ possesses an atom at $\mu$.



\section{Asymptotic expansion}
\label{sec:expansion}

Stationarity of the time series $(X_i)_{i \ge 1}$ means that for all positive integers $k \ge 1$ and $h \ge 0$, the distribution of the random vector $(X_{1+h}, \ldots, X_{k+h})$ does not depend on $h$. By the Birkhoff ergodic theorem \citep[Theorem~10.6]{kallenberg:2002}, stationarity and ergodicity imply that, for every Borel measurable function $f : \reals \to \reals$ such that $\expec[\abs{f(X_1)}] < \infty$, we have
\begin{equation}
\label{eq:ergodic}
  \frac{1}{n} \sum_{i=1}^n f(X_i) \to \expec[ f(X_1) ], \qquad n \to \infty, \text{ almost surely.}
\end{equation}

Assume that the stationary distribution $F$ has finite mean, $\mu$, and recall the mean absolute deviation $\theta$ in \eqref{eq:MAD} and its sample version $\hat{\theta}_n$ in \eqref{eq:MADn}. Write
\begin{equation}
\label{eq:MAD:easy}
  \hat{\MAD}_n - \MAD = \frac{1}{n} \sum_{i=1}^n \bigl( |X_i - \bar{X}_n| - |X_i - \mu| \bigr)
  + \frac{1}{n} \sum_{i=1}^n (\abs{X_i - \mu} - \MAD)
\end{equation}
The second term on the right-hand side is just a sum of centered random variables. It is the first term which poses a challenge.



\begin{lemma}
\label{lem:infl}
Let $X_1, X_2, \ldots$ be a stationary, ergodic time series with finite mean $\mu = \expec[X_1]$. We have, as $n \to \infty$, almost surely,
\begin{align*}
  \frac{1}{n} \sum_{i = 1}^n \bigl( |X_i - \bar{X}_n| - |X_i - \mu| \bigr)
  &= (\bar{X}_n - \mu) \, \bigl( \Pr[X_1 < \mu] - \Pr[X_1 > \mu] \bigr) \\
  & \qquad \mbox{}
  + \abs{\bar{X}_n - \mu} \, \Pr[ X_1 = \mu ] + o \bigl( \abs{\bar{X}_n - \mu} \bigr).
\end{align*}
\end{lemma}

\begin{proof}
Consider the random variables
\begin{align*}
  A_n &= \min( \bar{X}_n, \mu ), & B_n &= \max( \bar{X}_n, \mu ).
\end{align*}
Let $\sign(z)$ be equal to $1$, $0$, or $-1$ according to whether $z$ is larger than, equal to, or smaller than zero, respectively. For $x \in \RR \setminus (A_n, B_n)$, that is, for $x$ not between $\mu$ and $\bar{X}_n$, a case-by-case analysis reveals that
\[
  \abs{x - \bar{X}_n} - \abs{x - \mu}
  = ( \bar{X}_n - \mu ) \, \sign( \mu - x ) + \abs{\bar{X}_n - \mu} \, \1_{\{\mu\}}(x).
\]
Consider the partition $\{1, \ldots, n\} = \mathcal{K}_n \cup \mathcal{L}_n$, where
\begin{align*}
  \mathcal{K}_n &= \{ i = 1, \ldots, n : A_n < X_i < B_n \}, \\
  \mathcal{L}_n &= \{1, \ldots, n\} \setminus \mathcal{K}_n.
\end{align*}
By convention, the sum over the empty set is zero. We find that
\begin{align*}
  \lefteqn{
  \sum_{i = 1}^n \bigl( \abs{X_i - \bar{X}_n} - \abs{X_i - \mu} \bigr)
  } \\
  &= (\bar{X}_n - \mu) \, \sum_{i \in \mathcal{L}_n } \sign( \mu - X_i ) + \abs{\bar{X}_n - \mu} \, \sum_{i \in \mathcal{L}_n } \1_{\{\mu\}}(X_i) \\
  &\quad\mbox{} + \sum_{i \in \mathcal{K}_n} \bigl( \abs{X_i - \bar{X}_n} - \abs{X_i - \mu} \bigr) \\
  &= (\bar{X}_n - \mu) \, \sum_{i = 1}^n \sign( \mu - X_i ) + \abs{ \bar{X}_n - \mu } \, \sum_{i = 1}^n \1_{\{\mu\}}(X_i) + R_n
\end{align*}
with
\begin{multline*}
  R_n 
  = \sum_{i \in \mathcal{K}_n} \bigl( \abs{X_i - \bar{X}_n} - \abs{X_i - \mu} \bigr) \\
  - (\bar{X}_n - \mu) \, \sum_{i \in \mathcal{K}_n } \sign( \mu - X_i ) 
  - \abs{ \bar{X}_n - \mu } \, \sum_{i \in \mathcal{K}_n } \1_{\{\mu\}}(X_i)
\end{multline*}
By \eqref{eq:ergodic}, we have, as $n \to \infty$, almost surely,
\begin{align*}
  \frac{1}{n} \sum_{i = 1}^n \sign( \mu - X_i )
  &\to \expec[ \sign( \mu - X_1 ) ] = \Pr[X_1 < \mu] - \Pr[X_1 > \mu], \\
  \frac{1}{n} \sum_{i = 1}^n \1_{\{\mu\}}(X_i)
  &\to \Pr[ X_1 = \mu ].
\end{align*}
As a consequence, as $n \to \infty$ and almost surely,
\begin{align*}
  \frac{1}{n} \sum_{i = 1}^n \bigl( \abs{X_i - \bar{X}_n} - \abs{X_i - \mu} \bigr) 
  &= ( \bar{X}_n - \mu ) \bigl( \Pr[X_1 < \mu] - \Pr[X_1 > \mu] + o(1) \bigr) \\
  &\qquad \mbox{} + \abs{ \bar{X}_n - \mu } \, \bigl( \Pr[ X_1 = \mu ] + o(1) \bigr) + \frac{R_n}{n}.
\end{align*}
Further, as $\abs{ \abs{x - \bar{X}_n} - \abs{x - \mu} } \le \abs{ \bar{X}_n - \mu }$ for all $x \in \reals$, we obtain
\[
  \abs{R_n} \le 3 \abs{ \bar{X}_n - \mu } \, \abs{ \mathcal{K}_n },
\]
where $\abs{ \mathcal{K}_n }$ is the number of elements of $\mathcal{K}_n$. The lemma therefore follows if we can prove that
\begin{equation}
\label{eq:Kn0}
  \frac{\abs{ \mathcal{K}_n }}{n} \to 0, \qquad n \to \infty, \text{ a.s.}
\end{equation}
For $x \in \reals$, put
\begin{align*}
  \hat{F}_n(x) &= \frac{1}{n} \sum_{i=1}^n \1(X_i \le x), &
  \hat{F}_n(x-) &= \frac{1}{n} \sum_{i=1}^n \1(X_i < x).
\end{align*}
We have
\[
  \frac{\abs{ \mathcal{K}_n }}{n} = \hat{F}_n( B_n- ) - \hat{F}_n( A_n ).
\]
By \eqref{eq:ergodic}, $\bar{X}_n$ and thus $A_n$ and $B_n$ converge to $\mu$ almost surely. Moreover, we have either $A_n = \mu \le B_n$ or $A_n \le \mu = B_n$ (or both, if $\bar{X}_n = \mu$). As a consequence,
\[
  \frac{\abs{ \mathcal{K}_n }}{n}
  \le \max \bigl\{ \hat{F}_n( B_n - ) - \hat{F}_n(\mu), \hat{F}_n( \mu - ) - \hat{F}_n(A_n) \bigr\}.
\]
Fix $\delta > 0$. By monotonocity of $\hat{F}_n$ and by ergodicity \eqref{eq:ergodic}, we have, almost surely,
\begin{align*}
  \limsup_{n \to \infty} \hat{F}_n( B_n - ) 
  &\le \limsup_{n \to \infty} \hat{F}_n( \mu + \delta ) = F(\mu + \delta), \\
  \liminf_{n \to \infty} \hat{F}_n( A_n ) 
  &\ge \liminf_{n \to \infty} \hat{F}_n( \mu - \delta ) = F(\mu - \delta).
\end{align*}
Since moreover $\hat{F}_n(\mu) \to F(\mu)$ and $\hat{F}_n(\mu-) \to F(\mu-) = \Pr[X_1 < \mu]$ almost surely, it follows that
\begin{align*}
  \limsup_{n \to \infty} \bigl( \hat{F}_n( B_n - ) - \hat{F}_n(\mu) \bigr) 
  &\le F(\mu + \delta) - F(\mu), \\
  \limsup_{n \to \infty} \bigl( \hat{F}_n( \mu - ) - \hat{F}_n(A_n) \bigr) 
  &\le F( \mu - ) - F(\mu - \delta),
\end{align*}
almost surely, and therefore
\[
  \limsup_{n \to \infty} \frac{| \mathcal{K}_n |}{n} \le
  \max \bigl( F(\mu + \delta) - F(\mu), \, F( \mu - ) - F(\mu - \delta) \bigr)
\]
almost surely. Since $\delta$ was arbitrary, we obtain \eqref{eq:Kn0}, as required.
\end{proof}

\section{Weak convergence}
\label{sec:dto}

Combining the expansions in \eqref{eq:MAD:easy} and Lemma~\ref{lem:infl}, the limit distribution of $\hat{\MAD}_n$ follows right away. For the sake of illustration, we consider two settings: strongly mixing time series with finite second moments (Subsection~\ref{ss:serial}) and independent random sampling from distributions in the domain of attraction of a stable law with index $\alpha \in (1, 2)$ (Subsection~\ref{ss:stable}).

\subsection{Strongly mixing time series, finite variance}
\label{ss:serial}

Let $(X_i)_{i \in \ZZ}$ be a stationary time series. For $k \in \ZZ$, consider the $\sigma$-fields $\mathcal{F}_k = \sigma(X_i : i \le k)$ and $\mathcal{G}_k = \sigma(X_i : i \ge k)$. Rosenblatt's mixing coefficients are defined by
\[
  \alpha_n = 
  \sup 
  \bigl\{ 
    \abs{ \Pr(A \cap B) - \Pr(A) \, \Pr(B) } : 
    A \in \mathcal{F}_0, \, B \in \mathcal{G}_n
  \bigr\}
\]
for integer $n \ge 0$. The time series $(X_i)_{i \in \ZZ}$ is called strongly mixing if $\alpha_n \to 0$ as $n \to \infty$. Strong mixing implies ergodicity.

Define $\alpha(t) = \alpha_{\floor{t}}$, $t \ge 0$, and its inverse function $\alpha^{-1}(u) = \inf \{ t \ge 0 : \alpha(t) \le u \}$ for $u > 0$. Let $Q$ denote the quantile function of the distribution of $X_i$. 

\begin{proposition}
\label{prop:serial}
Let $(X_i)_{i \in \ZZ}$ be a strongly mixing, stationary time series with finite second moments. If
\begin{equation}
\label{eq:alphaQ}
  \int_0^1 \alpha^{-1}(u) \, \{Q(u)\}^2 \, \diff u < \infty,
\end{equation}
then, as $n \to \infty$,
\begin{equation}
\label{eq:MAD:limit}
  \sqrt{n} ( \MADn - \MAD ) 
  \dto 
  Y \, \bigl( \Pr[ X_1 < \mu ] - \Pr[ X_1 > \mu ] \bigr) 
  + \abs{Y} \, \Pr[ X_1 = \mu ] 
  + Z,
\end{equation}
where $(Y, Z)$ is bivariate normal with mean zero and covariance matrix given by
\begin{align*}
  \var(Y) 
  &= \var(X_0) + 2 \, \sum_{i=1}^\infty \cov(X_i, X_0), \\
  \var(Z) 
  &= \var(\abs{X_0 - \mu}) + 2 \, \sum_{i=1}^\infty \cov(\abs{X_i - \mu}, \abs{X_0 - \mu}), \\
  \cov(Y, Z) 
  &= \sum_{i \in \ZZ} \cov(\abs{X_i - \mu}, X_0) 
  = \sum_{i \in \ZZ} \cov(X_i, \abs{X_0 - \mu}).
\end{align*}
all series being absolutely convergent.
\end{proposition}

\begin{proof}
By Theorem~1 in \cite{doukhan:massart:rio:1994} and the remark just after that theorem, we have, as $n \to \infty$,
\begin{equation}
\label{eq:YnZn}
  (Y_n, Z_n) :=
  \left( 
    \frac{1}{\sqrt{n}} \sum_{i=1}^n (X_i - \mu), \,
    \frac{1}{\sqrt{n}} \sum_{i=1}^n \bigl( \abs{X_i - \mu} - \theta \bigr)
  \right)
  \dto (Y, Z),
\end{equation}
with $(Y, Z)$ as in the statement of the proposition. Combining \eqref{eq:MAD:easy} and Lemma~\ref{lem:infl} yields the expansion
\begin{multline}
\label{eq:MAD:YnZn}
  \sqrt{n} ( \MADn - \MAD ) \\
  = Y_n \, \bigl( \Pr[X_1 < \mu] - \Pr[X_1 > \mu] \bigr)
  + \abs{Y_n} \, \Pr[ X_1 = \mu ] + o( \abs{Y_n} ) + Z_n,
\end{multline}
as $n \to \infty$, almost surely. Apply the continuous mapping theorem and Slutsky's lemma to arrive at the result.
\end{proof}

Condition~\ref{eq:alphaQ} covers many cases and is almost sharp; see the applications on pages~67--68 in \cite{doukhan:massart:rio:1994} as well as their Theorem~2. The case of independent and identically distributed variables is trivially included.

\begin{corollary}
\label{cor:CLT}
Let $X_1, X_2, \ldots$ be independent and identically distributed random variables with common distribution $F$. If $F$ has a finite second moment, then weak convergence \eqref{eq:MAD:limit} holds, where $(Y, Z)$ is bivariate normal with mean zero and covariance matrix equal to the one of $( X_1 - \mu, \abs{ X_1 - \mu } )$.
\end{corollary}

\begin{proof}
Combine \eqref{eq:YnZn} and \eqref{eq:MAD:YnZn} with the multivariate central limit theorem, the continuous mapping theorem, and Slutsky's lemma.
\end{proof}

If $\Pr[X_1 = \mu] = 0$, i.e., if $F$ is continuous at $0$, then the expansion \eqref{eq:MAD:YnZn} is the same as the one in \eqref{eq:MAD:empproc} and we obtain that $\sqrt{n} ( \MADn - \MAD )$ is asymptotically normal. In the independence case, the limit distribution coincides with the one in \eqref{eq:MAD:asnorm}. However, if $0 < \Pr[ X_1 = \mu ] < 1$, then the weak limit in \eqref{eq:MAD:limit} is not Gaussian. Moreover, the limit distribution is not centered either, its expectation being $\Pr[ X_1 = \mu ] \, \expec[\abs{Y}]$.


\subsection{Independent random sampling, infinite variance}
\label{ss:stable}

One argument in favour of the use of the mean absolute deviation for measuring dispersion is that, unlike the standard deviation, it does not require existence of second moments. Still, in the weak convergence statements \eqref{eq:MAD:asnorm} and \eqref{eq:MAD:limit}, finite second moments are presupposed. This condition is lifted in the next result.  A positive, Borel measurable function $L$ defined on a neighbourhood of infinity is \emph{slowly varying} if $\lim_{x \to \infty} L(xy) / L(x) = 1$ for all $y > 0$.




\begin{proposition}
\label{prop:asym:stable}
Let $X_1, X_2, \ldots$ be independent, identically distributed random variables with common distribution function $F$. Assume that there exist $\alpha \in (1, 2)$, $p \in [0, 1]$, and a slowly varying function $L$ on $(0, \infty)$, such that, for $x > 0$,
\begin{equation}
\label{eq:RV}
  \left.
  \begin{array}{rcl}
  \Pr[ X_1 > x ] &=& p \, x^{-\alpha} \, L(x), \\[1ex]
  \Pr[ X_1 < -x ] &=& (1-p) \, x^{-\alpha} \, L(x).
  \end{array}
  \right\}
\end{equation}
Then $F$ has a finite mean but an infinite second moment. Defining
\[
  a_n = \inf \{ x > 0 : \Pr[ \abs{X_1} > x ] \le 1/n \},
\]
we have, assuming that $F$ is continuous at $\mu$,
\[
  \frac{n}{a_n} ( \MADn - \MAD ) \dto G, \qquad n \to \infty,
\]
where $G$ is a stable distribution with characteristic function
\begin{equation}
\label{eq:stable:G}
  \int_{\reals} e^{-\mathrm{i} s x} \, \diff G(x) 
  = \exp 
  \bigl\{ 
    - \sigma^\alpha \abs{s}^\alpha \bigl( 1 - \mathrm{i} \sign(s) \tan(\alpha \pi/2) \bigr) 
  \bigr\},
  \qquad s \in \reals,
\end{equation}
the scale parameter $\sigma > 0$ being given by
\begin{equation}
\label{eq:stable:sigma}
  \sigma^\alpha = 
  \frac%
  {2^\alpha \{\Pr[ X_1 < \mu ]^\alpha \, p + \Pr[ X_1 > \mu]^\alpha \, (1-p)\}}%
  {\frac{\Gamma(2-\alpha)}{\alpha - 1} \abs{ \cos( \alpha \pi / 2 ) }}.
\end{equation}
%
\end{proposition}

\begin{proof}
The statement about the existence of the moments is well known; see for instance the unnumbered lemma on page~578 in Section~XVII.5 in \cite{feller:1971}. Since $\Pr[ X_1 = \mu ] = 0$, equation~\eqref{eq:MAD:easy} and Lemma~\ref{lem:infl} imply
\begin{equation}
\label{eq:MAD:xi}
  \MADn - \MAD 
  = \frac{1}{n} \sum_{i=1}^n (\xi_i - \theta) + o \bigl( \abs{\bar{X}_n - \mu} \bigr),
  \qquad n \to \infty, \text{ a.s.},
\end{equation}
where, writing $b = \Pr[ X_1 < \mu ] - \Pr[ X_1 > \mu ]$,
\begin{equation}
\label{eq:xi}
  \xi_i 
  = \abs{X_i - \mu} + b \, (X_i - \mu) 
  = \abs{X_i - \mu} \, \bigl(1 + b \sign(X_i - \mu) \bigr).
\end{equation}
Note that $\MAD = \expec[\xi_i]$. Since $-1 < b < 1$, we have $\xi_i \ge 0$. Moreover, by slow variation of $L$, as $x \to \infty$,
\begin{align}
\nonumber
  \Pr[ \xi_i > x ] 
  &= \Pr[ X_i > \mu + x / (1+b) ] + \Pr[ X_i < \mu - x / (1 - b) ] \\
\label{eq:xi:L}
  &= \bigl((1+b)^\alpha \, p + (1-b)^\alpha \, (1-p) + o(1) \bigr) \, x^{-\alpha} \, L(x).
\end{align}
As $\Pr[ \abs{X_i} > x ] = x^{-\alpha} \, L(x)$, it follows that, up to a multiplicative constant, the tail function of $\xi_i$ is asymptotically equivalent to the one of $\abs{X_i}$. Observe that $1+b = 2 \, \Pr[ X_1 < \mu ]$ and $1-b = 2 \, \Pr[ X_1 > \mu ]$.

By classical theory on the domains of attraction of non-Gaussian stable distributions, equation~\eqref{eq:RV} is equivalent to the weak convergence of $n a_n^{-1} (\bar{X}_n - \mu) = a_n^{-1} \sum_{i=1}^n (X_i - \mu)$ to an $\alpha$-stable distribution \citep{gnedenko:kolmogorov:1954}. The tails of $\abs{X_i}$ and $\xi_i$ being related through~\eqref{eq:xi:L}, $a_n^{-1} \sum_{i=1}^n (\xi_i - \theta)$ converges weakly to an $\alpha$-stable distributions too. The limit distribution can be found for instance from Theorem~1.8.1 in \cite{samorodnitsky:taqqu:1994} and coincides with $G$ in the statement of the proposition; details are given below. By \eqref{eq:MAD:xi}, we have
\[
  \frac{n}{a_n} \bigl( \MADn - \MAD \bigr)
  = \frac{1}{a_n} \sum_{i=1}^n (\xi_i - \theta) + o_p(1), \qquad n \to \infty.
\]
Weak convergence of $n a_n^{-1} \bigl( \MADn - \MAD \bigr)$ to $G$ now follows from Slutsky's lemma.

The calculations leading to the expression of the scale parameter $\sigma$ in equation~\eqref{eq:stable:sigma} are as follows. Let $\tilde{a}_n$ be such that
\begin{equation}
\label{eq:tailxi:1}
  \lim_{n \to \infty} n \, \Pr[ \xi_i > \tilde{a}_n ]
  = \frac{\Gamma(2-\alpha)}{\alpha - 1} \abs{ \cos( \alpha \pi / 2 ) }.
\end{equation}
By Theorem~1.8.1 in \cite{samorodnitsky:taqqu:1994}, we have
\[
  \frac{1}{\tilde{a}_n} \sum_{i=1}^n (\xi_i - \theta) \dto \tilde{G}, \qquad n \to \infty,
\]
where $\tilde{G}$ is an $\alpha$-stable distribution whose characteristic function has the same form as in \eqref{eq:stable:G} with $\sigma$ replaced by $\tilde{\sigma} = 1$. By \eqref{eq:xi:L} and the definition of $a_n$, we have
\begin{equation}
\label{eq:tailxi:2}
  \lim_{n \to \infty} n \, \Pr[ \xi_i > a_n ]
  = (1+b)^\alpha \, p + (1-b)^\alpha \, (1-p).
\end{equation}
The function $x \mapsto \Pr[ \xi_i > x ]$ being regularly varying with index $-\alpha$, it follows that a valid choice for $\tilde{a}_n$ in \eqref{eq:tailxi:1} is $\tilde{a}_n = \gamma a_n$, where the constant $\gamma > 0$ can be read off by comparing \eqref{eq:tailxi:1} and \eqref{eq:tailxi:2}:
\begin{align*}
  \gamma^{-\alpha} \, \{ (1+b)^\alpha \, p + (1-b)^\alpha \, (1-p) \}
  &= \lim_{n \to \infty} n \, \Pr[ \xi_i > \gamma a_n ] \\
  &= \frac{\Gamma(2-\alpha)}{\alpha - 1} \abs{ \cos( \alpha \pi / 2 ) }.
\end{align*}
If $\tilde{Z}$ denotes a random variable whose distribution function is $\tilde{G}$, then
\[
  \frac{1}{a_n} \sum_{i=1}^n (\xi_i - \theta)
  = \frac{\gamma}{\tilde{a}_n} \sum_{i=1}^n (\xi_i - \theta)
  \dto \gamma \tilde{Z} = Z, \qquad n \to \infty.
\]
By \citet[Property~1.2.3]{samorodnitsky:taqqu:1994}, the characteristic function of the law of $Z$ is given by \eqref{eq:stable:G} with scale parameter $\sigma = \gamma \tilde{\sigma} = \gamma$.
\end{proof}

The regular variation condition~\eqref{eq:RV} covers distributions with power-law tails, such as the Pareto distribution, non-Gaussian stable distributions, the Student t distribution, and the Fr\'echet distribution.

\begin{remark}
Proposition~\ref{prop:asym:stable} supposes independent random sampling from a heavy-tailed distribution. Extensions to weakly dependent stationary time series are possible too. In \cite{bartkiewicz:etal:2011}, for instance, conditions are given which guarantee the weak convergence of the normalized partial sums of a weakly dependent stationary time series to a non-Gaussian stable distribution. These conditions are to be verified for the sequence $(\xi_i)_{i}$ defined in \eqref{eq:xi}. The expansion in \eqref{eq:MAD:xi} then allows to obtain the asymptotic distribution of the sample mean absolute deviation.
\end{remark}

\begin{remark}
In Proposition~\ref{prop:asym:stable}, if $F$ is not continuous at $\mu$, then $n a_n^{-1} ( \MADn - \MAD )$ can still be shown to converge weakly, but the limit law will no longer be stable. As in Corollary~\ref{cor:CLT}, it will be a non-linear functional of the bivariate stable distribution to which the joint distribution of $(X_1 - \mu, \abs{X_1 - \mu})$ is attracted \citep{rvaceva:1962}.
\end{remark}

%
%
%

\section*{Acknowledgments}

The author gratefully acknowledges funding by contract ``Projet d'Act\-ions de Re\-cher\-che Concert\'ees'' No.\ 12/17-045 of the ``Communaut\'e fran\c{c}aise de Belgique'' and by IAP research network Grant P7/06 of the Belgian government (Belgian Science Policy).

\bibliographystyle{chicago}
\bibliography{biblio}

\def\cprime{$'$}
\begin{thebibliography}{}

\bibitem[\protect\citeauthoryear{Babu and Rao}{Babu and
  Rao}{1992}]{babu:rao:1992}
Babu, G.~J. and C.~R. Rao (1992).
\newblock Expansions for statistics involving the mean absolute deviations.
\newblock {\em Ann. Inst. Statist. Math.\/}~{\em 44\/}(2), 387--403.

\bibitem[\protect\citeauthoryear{Bartkiewicz, Jakubowski, Mikosch, and
  Wintenberger}{Bartkiewicz et~al.}{2011}]{bartkiewicz:etal:2011}
Bartkiewicz, K., A.~Jakubowski, T.~Mikosch, and O.~Wintenberger (2011).
\newblock Stable limits for sums of dependent infinite variance random
  variables.
\newblock {\em Probab. Theory Related Fields\/}~{\em 150\/}(3-4), 337--372.

\bibitem[\protect\citeauthoryear{Doukhan, Massart, and Rio}{Doukhan
  et~al.}{1994}]{doukhan:massart:rio:1994}
Doukhan, P., P.~Massart, and E.~Rio (1994).
\newblock The functional central limit theorem for strongly mixing processes.
\newblock {\em Ann. Inst. H. Poincar\'e Probab. Statist.\/}~{\em 30\/}(1),
  63--82.

\bibitem[\protect\citeauthoryear{Feller}{Feller}{1971}]{feller:1971}
Feller, W. (1971).
\newblock {\em An Introduction to Probability Theory and Its Applications.
  {V}ol. {II}.}
\newblock Second edition. John Wiley \& Sons, Inc., New York-London-Sydney.

\bibitem[\protect\citeauthoryear{Gnedenko and Kolmogorov}{Gnedenko and
  Kolmogorov}{1954}]{gnedenko:kolmogorov:1954}
Gnedenko, B.~V. and A.~N. Kolmogorov (1954).
\newblock {\em Limit Distributions for Sums of Independent Random Variables}.
\newblock Addison-Wesley Publishing Company, Inc., Cambridge, Mass.
\newblock Translated and annotated by K. L. Chung. With an Appendix by J. L.
  Doob.

\bibitem[\protect\citeauthoryear{Gorard}{Gorard}{2005}]{gorard:2005}
Gorard, S. (2005).
\newblock Revisiting a 90-year-old debate: the advantages of the mean
  deviation.
\newblock {\em British Journal of Educational Studies\/}~{\em 53\/}(4),
  417--430.

\bibitem[\protect\citeauthoryear{Huber}{Huber}{1996}]{huber:1996}
Huber, P.~J. (1996).
\newblock {\em Robust Statistical Procedures\/} (Second ed.), Volume~68 of {\em
  CBMS-NSF Regional Conference Series in Applied Mathematics}.
\newblock Society for Industrial and Applied Mathematics (SIAM), Philadelphia,
  PA.

\bibitem[\protect\citeauthoryear{Kallenberg}{Kallenberg}{2002}]{kallenberg:2002}
Kallenberg, O. (2002).
\newblock {\em Foundations of Modern Probability\/} (Second ed.).
\newblock Probability and its Applications (New York). Springer-Verlag, New
  York.

\bibitem[\protect\citeauthoryear{Muñoz-Perez and Sanchez-Gomez}{Muñoz-Perez
  and Sanchez-Gomez}{1990}]{munoz:sanchez:1990}
Muñoz-Perez, J. and A.~Sanchez-Gomez (1990).
\newblock A characterization of the distribution function: the dispersion
  function.
\newblock {\em Statistics \& Probability Letters\/}~{\em 10\/}(3), 235--239.

\bibitem[\protect\citeauthoryear{Pham-Gia and Hung}{Pham-Gia and
  Hung}{2001}]{pham-gia:hung:2001}
Pham-Gia, T. and T.~L. Hung (2001).
\newblock The mean and median absolute deviations.
\newblock {\em Math. Comput. Modelling\/}~{\em 34\/}(7-8), 921--936.

\bibitem[\protect\citeauthoryear{Pollard}{Pollard}{1989}]{pollard:1989}
Pollard, D. (1989).
\newblock Asymptotics via empirical processes.
\newblock {\em Statistical Science\/}~{\em 4\/}(4), pp. 341--354.

\bibitem[\protect\citeauthoryear{Rva{\v{c}}eva}{Rva{\v{c}}eva}{1962}]{rvaceva:1962}
Rva{\v{c}}eva, E.~L. (1962).
\newblock On domains of attraction of multidimensional distributions.
\newblock {\em Selected Translations in Math. Statist. Prob. Theory\/}~{\em 2},
  183--205.

\bibitem[\protect\citeauthoryear{Samorodnitsky and Taqqu}{Samorodnitsky and
  Taqqu}{1994}]{samorodnitsky:taqqu:1994}
Samorodnitsky, G. and M.~S. Taqqu (1994).
\newblock {\em Stable non-{G}aussian Random Processes}.
\newblock Stochastic Modeling. Chapman \& Hall, New York.
\newblock Stochastic models with infinite variance.

\bibitem[\protect\citeauthoryear{Stigler}{Stigler}{1973}]{stigler:1973}
Stigler, S.~M. (1973).
\newblock Studies in the history of probability and statistics. {XXXII}.
  {L}aplace, {F}isher, and the discovery of the concept of sufficiency.
\newblock {\em Biometrika\/}~{\em 60}, 439--445.

\bibitem[\protect\citeauthoryear{Tukey}{Tukey}{1960}]{tukey:1960}
Tukey, J.~W. (1960).
\newblock A survey of sampling from contaminated distributions.
\newblock In {\em Contributions to probability and statistics}, pp.\  448--485.
  Stanford Univ. Press, Stanford, Calif.

\bibitem[\protect\citeauthoryear{van~der Vaart}{van~der
  Vaart}{1998}]{vdvaart:1998}
van~der Vaart, A.~W. (1998).
\newblock {\em Asymptotic Statistics}.
\newblock Cambridge: Cambridge University Press.

\end{thebibliography}

\end{document}